\documentclass[12pt,titlepage]{article}

\usepackage[hyphens]{url}
\usepackage{lipsum}
\usepackage{graphicx,color}

\usepackage{todonotes}

\usepackage[utf8]{inputenc}
\usepackage[T1]{fontenc}
\usepackage[cmex10]{amsmath}
\usepackage{amsthm}
\usepackage{amssymb}
\usepackage{mathabx}
\usepackage{setspace}
\usepackage[OT1]{fontenc}
\usepackage{tikz}
\usetikzlibrary{arrows}
\usepackage{caption}
\usetikzlibrary{matrix}
\usepackage{amsmath}
\usepackage{float}

\usepackage{subcaption}

\makeatletter
\renewcommand\paragraph{\@startsection{paragraph}{4}{\z@}%
  {-3.25ex \@plus -1ex \@minus -0.2ex}%
  {0.01pt}%
  {\normalfont\bfseries\nobreak}%
}
\makeatother

\def\one{\mathbf{1}}

\def\one{\mathbf{1}}

\def\Re{\mathbf{R}}
\def\os{\emptyset}

\newcommand{\df}[1]{\textbf{\textit{#1}}}

\def\succsim{\succeq}
\def\kpref{\trianglerighteq}
\def\ksim{\thickapprox}

\theoremstyle{plain}
\newtheorem{theorem}{Theorem}
\newtheorem{lemma}[theorem]{Lemma}
\newtheorem{proposition}[theorem]{Proposition}
\newtheorem{corollary}[theorem]{Corollary}

\theoremstyle{definition}
\newtheorem{definition}{Definition}
\newtheorem{example}{Example}

\theoremstyle{remark}
\newtheorem*{remark}{Remark}

\usepackage[
    backend=biber,natbib,
    style=authoryear,
    citetracker=true,
    maxcitenames=1
]{biblatex}

\AtEveryCitekey{\ifciteseen{}{\defcounter{maxnames}{99}}}

\begin{document}

\title{A Mathematical Theory of Classifiers;\\ Representations and Applications\thanks{
The first author gratefully acknowledges support
from Beyond Limits (Learning Optimal Models) through CAST (The Caltech Center for Autonomous Systems and Technologies). The second author gratefully acknowledges support from the National Science Foundation under grant number SES-1558757. We thank Chris Chambers for many comments and suggestions. We would like to thank our reviewers for their constructive feedback.
}}

\author{Hamed Hamze Bajgiran \\ Caltech \and Federico Echenique \\
  Caltech }

\maketitle

\begin{abstract}
We study the complexity of closure operators, with applications to machine learning and decision theory. In machine learning, closure operators emerge naturally in data classification and clustering. In decision theory, they can model equivalence of choice menus, and therefore situations with a preference for flexibility. Our contribution is to formulate a notion of complexity of closure operators, which translate into the complexity of a classifier in ML, or of a utility function in decision theory. 
\end{abstract}

\section{Introduction}

We study the complexity of closure operators, with applications to machine learning and decision theory in economics. Closure operators can be used to model a coarsening of reality. For example, the set of all 16 subsets of the set $\{a,b,c,d\}$ may be viewed as a fine-grained description of some physical system. A coarser view could regard some subsets as equivalent, just like in the usual topology on the real line, the rational numbers are ``indistinguishable'' from the real numbers in the sense that their closure equals the reals. For example, all subsets of $\{a,c\}$ in $\{a,b,c,d\}$ could be seen as equivalent, which a closure operator may capture by mapping the sets $\{a\}$, $\{c\}$, and $\{a,c\}$ into $\{a,c\}$. 

Such coarsening of reality emerges naturally in machine learning and in decision theory. In machine learning, a \textit{classifier} will decide that some individual data-points should be lumped together. For example, the classifier can be applied to images of household pets, and classify them according to species, without regard for color or size. An Irish Setter is  indistinguishable from a Dachshund: both are mapped into the class of dogs; while a canary and a parrot may be seen as distinct from any dogs, but grouped together as birds.

In economic decision theory, an agent may be choosing a menu, or a store, from which to make an ultimate choice. Someone who is in the market for a  particular item, say a Canon 5D Mark III,  will regard any photography store that carries this model as equivalent: one store indistinguishable from another as long as they have what the consumer is looking for. An agent choosing items in  $\{a,b,c,d\}$, with preferences $a\succ b\succ c\succ d$, will consider all menus that contain $a$ as indifferent; any menu that contains $b$ but not $a$ is also indifferent, and so on. 

Closure operators are ideally suited to analyze these situations. In particular, we shall see that closure operators can capture a desire for flexibility stemming from agents who are uncertain about their preferences. If our agent is unsure of whether they rank $a$ over $b$, or the other way around, they may have a strict preference for a larger menu containing both options. For example they may prefer $\{a\}$ over $\{b\}$ (as menus), but strictly prefer $\{a,b\}$ over $\{a\}$. A branch of decision theory (one that started with the seminal paper of \cite{kreps}) relies on closure operators as a modeling device. In particular, the work of \cite*{DLR} uses the operator that maps a set $A\subseteq\Re^n$ to its convex hull; and their key axiom is that agents are indifferent between $A$ and its convex hull.

Our starting point is a representation theorem for closure operators. We shall see that a closure operator may be built up from more basic, or simpler, operators. In machine learning, it is natural to think of a classifier as resulting from the composition of simple linear classifiers. We show that a similar representation is possible, even in discrete settings where there is no language for talking about linearity. In decision theory, some agents may have complex preferences over menus that result from a desire for flexibility in choice. We shall see that such complex preferences can be constructed from  simpler operators that do not exhibit such a desire for flexibility. 
 
Our representation theorem contains those of \cite*{richterandrubinstein} and \cite*{chambers2020closure} as special cases, and relies on similar ideas. The main novelty in our work lies in connecting the topologies arising from a closure operator, to those arising from its simplest constituent operators; and in using this connection to formulate a notion of complexity.  It is indeed natural in machine learning to identify the complexity of a classifier with the number of different simple classifiers needed to represent it: our paper formalizes this connection in the abstract model of discrete data. 

For decision theory, we shall see that the complexity of an operator is tied to the size of the (subjective) state-space needed to represent an agent's choices. As we wrote above, a preference for flexibility emerges from uncertainty about future preferences. Such uncertainty can be represented with a state space that is obtained (as in \cite{kreps}) from the relevant operator. Indeed, \cite*{DLR} write:
\begin{quote}
``the size of the subjective state space provides a measure of the agent's uncertainty about future contingencies.''
\end{quote}
We show how to extend the approach in Dekel et.\ al, so that if their key axiom is applied to an arbitrary closure operator, even in the discrete setting of our paper, a representation with a preference for flexibility may be obtained. In this representation, we can bound the size of the state-space through the representation of a closure operator in terms of simpler constituents.

\section{Literature Review}
Closure operators and their connections to topologies are a subject of study in mathematics (\cite{ward}), but the literature that is closest to our work pertains to applications in economics and computer science.  Some of these applications are related to convex geometries and their representations.

The basic concepts of abstract convex geometry and combinatorial convex hull operator are developed in \cite{edelmanandjamison}. Convex hull operators are a class of closure operators with an extra \emph{anti-exchange} property\footnote{We say that $f$ satisfies the anti-exchange property if given any closed set $A$ and two unequal points $x,y\in X\setminus A$, then $x\in f(A\cup y)$ implies that $y\notin f(A\cup x)$. }. In decision theory, \cite{koshevoy} studies the connection between the combinatorial convex hull operators and the path-independent choice functions.  The closest papers to ours are \cite{richterandrubinstein} and \cite{chambers2020closure}. \cite{richterandrubinstein} provide a characterization of a combinatorial convex hull operator through a set of primitive orderings. Using their representation, \cite{richterandrubinstein} propose a notion of competitive equilibrium in an abstract environment. \cite{chambers2020closure} extends the result of \cite{richterandrubinstein} and relates it to the \cite{kreps} model of preferences for flexibility in decision theory. They provide the main characterization of closure operators by weak orders.  

In the context of dynamic choice, and preferences for flexibility, following \cite{kreps}, \cite*{DLR}, and \cite*{GP}, a rather extensive literature in economics studies different aspects of choice over menus. The closest related papers are those of \cite{nehringspuppe1999}, \cite{kopylov1}, \cite{kopylov2}, and \cite{gorno}.  \cite{nehringspuppe1999} provides a more refined Pareto representation of the Krepsian model. They provide the connection to the representation by \cite{aizermanandmalishevski} of choice correspondences satisfying the heritage and outcast properties, which is another route for proving Claim 1 in \cite{richterandrubinstein}; that every convex geometry is generated by a set of primitive ordering.\footnote{In a nutshell, the heritage property states that if $A\subseteq B$, then $f(A)\supseteq f(B)\cap A$. When imposed on a choice function, this property is the $\alpha$ property of \cite{sen}. The outcast property states that if $f(A)\subseteq B\subseteq A$, then $f(A)=f(B)$. This property is a weaker form of property $\beta$ of \cite{sen}. Note, however, that our paper takes closure operators as primitive, not choice.} \cite{kopylov1} determines the number of positive and negative states in the setting of \cite*{DLR}. \cite{gorno} shows that any preference ordering in Kreps' setting has a representation as a Dekel-Lipman-Rustichini representation. Finally, \cite{kopylov2} proposes a combinatorial model of subjective states. By relaxing the axioms of Kreps, he presents a weaker model of coherent aggregation. 

Our application to classification problems in machine learning is related to the link between closure operators and Galois connections, and their applications in formal concept theory. Generally, there are different applications of lattice theory for studying general hierarchical concepts. For an overview of representations and applications to closure systems and formal concept theory see \cite{Caspard2003}, \cite{Domenach2004}, \cite{Bertet2018}, and \cite{Davey}. In our paper, we focus on a simple model with a new complexity notion, its decision-theoretic and conceptual interpretations, and how it may be  computed efficiently.

\section{Model and preliminaries}

\subsection{Preliminary definitions and notational conventions}

Let $X$ be a set. We denote the set of all subsets of $X$ by $2^X$. A binary relation on $X$ is 1) a \df{weak order} on $X$ if it is complete and transitive (meaning that it orders all elements of $X$); 2) a \df{partial order} if it is reflexive, transitive and anti-symmetric ($\forall x,y\in X$ if $x\succsim y$ and $y\succsim x$ then $x=y$); 3) a \df{total order} if it is  a complete partial order. 

Let $\succsim$ be a partial order on $X$. We say that two elements
$x,y$ with $x\succsim y$ or $y\succsim x$ are \df{ordered}, or
\df{comparable}, by $\succsim$. A subset $Y$ of $X$ is a \df{chain} if
it is totally ordered by $\succsim$, and an \df{anti-chain} if no two
of its elements are ordered by $\succsim$.  An element $x \in X$ is
said to be an \df{upper bound} of $Y$ if  $x\succsim a$ for every $a
\in Y$. We denote the \df{least upper bound}, or \df{join},  of $Y$
(if it exists) by $\bigvee Y$. Similarly, we denote the \df{greatest
  lower bound} or \df{meet} of $Y$ by $\bigwedge Y$. The partially
ordered set $X$ is a \df{lattice} if, for all $x,y\in X$, the set
$\{x,y\}$ has a join and a meet: denoted by $x\vee y$ and $x\wedge y$ respectively.

We introduce some concepts from discrete geometry that may be viewed as analogues to concepts in convex analysis: see \cite{richterandrubinstein} for an earlier application of these ideas to economics. Denote the set of all weak orders on $X$ by $\mathcal{R}$. The
\df{support function} of $A\subseteq X$ is the function
$h_A:\mathcal{R} \to 2^X$ defined by $h_A(\succsim)=\{x\in A|\ x\succsim
y\ \forall y\in A\}$. Similarly, we may define the \df{support
  half-space} of $A\neq\os$ with respect to $\succeq\in \mathcal{R}$
by $H(\succsim,h_A)=\{x\in X|\ h_A(\succsim)\succsim x\}$. By definition, we set $H(\succsim,h_{\emptyset})=\emptyset$.

Observe the analogy with convex analysis. Here the set $\mathcal{R}$ serves the same
role as the dual of $X$ (the set of continuous linear functionals),
when $X$ is a Euclidean space. 

One final notational convention is that we denote the indicator
function for a set $C\subseteq X$ by $\one_C$.  So $\one_C(x)=1$ if
$x\in C$ and $\one_C(x)=0$ otherwise.

\subsection{Closure operators}\label{sec:defnclosureoperator}
The paper is a study of a special kind of functions defined on the subsets of a finite set $X$; functions of the form $f:2^X \to 2^X$ that have the properties of closure operators in topology \citep{kuratowski2014topology}:

\begin{definition}
A \textbf{\emph{closure operator}} on $X$ is a map $f:2^X \to 2^X$ that satisfies the following properties:
\begin{enumerate}
    \item \textbf{\emph{Extensivity}}: $A\subseteq f(A)$ and $f(\emptyset)=\emptyset$.
    \item \textbf{\emph{Idempotence}}: $f(f(A))=f(A)$.
    \item \textbf{\emph{Monotonicity}}: $A\subseteq B $ implies $f(A)\subseteq f(B)$.\footnote{Kuratowski imposes $f(A\cup B)=f(A)\cup f(B)$, a stronger property than monotonicity.}
\end{enumerate}
\end{definition}

As applications of our theory, we shall discuss in detail two  concrete interpretations of closure operators below, but in the abstract one may think of the operator as a model of imperfect perception. The closure operator describes a coarse perception of reality: one that does not distinguish between $A$ and $f(A)$. The use of closure operators in topology follows this interpretation, as $f(A)$ consists of the points that are arbitrarily close to the elements of $A$.


The following are examples of closure operators:
\begin{enumerate}
\item The \df{identity operator} is the closure  operator $I:2^X\to2^X$ such that $I(A)=A$ for every $A\in 2^X$.
\item The \df{trivial closure operator} is defined as the closure operator $f:2^X\to2^X $ such that $f(A)=X$ for every nonempty $A\in 2^X$.
\item\label{it:binclass} A \df{binary classifier} is an operator $f_C$, associated to a set $C\subseteq X$. So that $f_C(\os)=\os$, and for nonempty $A$, $f_C(A)=C$ if $A\subseteq C$ and $f_C(A)=X$ otherwise.
\item\label{it:stratrat} Given a function $u:X\to\Re$, \[
f_u(A) = \{x\in X:u(x)\leq \max \{u(x'):x'\in A \} \}
  \] is the \df{strategically rational} operator defined from the function $u$.
 \end{enumerate}

Of the preceding examples, binary classifiers~\eqref{it:binclass} and strategically rational~\eqref{it:stratrat} operators will play an important role in the paper. It is worth discussing these two classes of closure operators in some detail before proceeding.

First note that a binary classifier gives rise to closed sets $S(f_C)=\{\os,C,X \}$, 
which are the smallest kind of non-trivial topology possible. In our paper, we shall think of these as simple classifiers. Binary classifiers are, moreover, a special case of strategically rational operators: Indeed for a given $C\subseteq X$ we may define $u=\one_{X\setminus C}$, and observe that $f_C=f_u$.

Turning to strategically rational operators, one may interpret $u:X\to\Re$ as a utility function, and $\max \{u(x):x\in A \}$ as the best utility achievable from a set $A\subseteq X$ of possible choices. Then $f_u(A)$ is the largest set of choices that an agent with utility $u$ would consider to be as good as choosing from $A$. In particular, if $f_u(A)=f_u(B)$ then the agent is equally happy choosing an item from $A$ as from $B$. 

The ``strategic rationality'' terminology is borrowed from \cite{kreps}, and will be useful when we talk about applications to decision theory. It suggests an agent who is forward looking, and identifies a menu with the element they plan to choose from the menu. Alternatively, strategically rational operators can be seen as abstract counterparts to the simple linear classifiers used in machine learning: we shall also emphasize this interpretation when we talk about applications to machine learning.

Finally note that all that matters about $u$ in the definition of $f_u$ is the weak order represented by $u$; that is,  $x\succeq y$ iff $u(x)\geq u(y)$. In other words, if $u=h\circ v$, for a strictly monotone increasing function $h:\Re\to\Re$, then $f_u=f_v$.  In the rest of the paper, we shall therefore write $f_\succeq$ for strategically rational closure operators, where it will be understood that $\succeq$ is a weak order over $X$. Observe that these coincide with the supporting half-spaces we introduced above.


\begin{definition}
A set $A\subseteq X$ is \df{closed} with respect to a closure operator $f:2^X \to 2^X$, if $A=f(A)$. The set $f(2^X)$ of all closed sets with respect to the closure operator $f$ is the \df{topology} defined by $f$, and is denoted by  $S(f)$.
\end{definition}

The terminology of closed sets and topology is justified by the following well-known result, which we state as Lemma~\ref{lem:closureintersection}.\footnote{In the context of convex geometry, which is most relevant for our paper, the result has been noticed in \cite{edelmanandjamison}. But the result is well known; see, for example, \cite{ward} and \cite{kuratowski2014topology}.}

\begin{lemma}\label{lem:closureintersection}
  Let $f:2^X\to 2^X$ be a closure operator on $X$, then the set of closed sets $S(f)$ is closed under intersection and contains $\emptyset$ and $X$. Indeed, $S(f)$ endowed with the meet and join operators $A\wedge B=A\cap B$ and  $A\vee B=f(A\cup B)$ is a lattice that has  $\emptyset$ and $X$ as its (respectively) smallest and largest elements. 

  Moreover, if $S$ is any subset of $2^X$ that is closed under intersection and contains $\emptyset$ and $X$, then there is a unique closure operator $f_S:2^X\to 2^X$ such that $S(f_S)=S$.
\end{lemma}

\subsection{Application 1: A theory of classifiers}
Interpret the elements of $X$ as data, and suppose given a finite set
$L$ of \df{labels}.  A \df{labeling correspondence} is a set-valued function $\Phi : X\rightrightarrows L$ that associates with each data point $x\in X$ a subset of labels $\Phi(x)$.

Given a labeling correspondence $\Phi : X\rightrightarrows L$, we define a \df{classifier} as a function $f:2^X\to 2^ X$ with $f(A)=\{x| x\in X, \ \bigcap_{y\in A} \Phi(y) \subseteq \Phi(x)\}$ for every $A\subseteq X$.

The interpretation of a labeling correspondence is straightforward. It attaches a set of labels to each data point. We interpret each label as a single feature or property attached to each data point. Hence, attaching two different labels $l_1,l_2\in L$ to a data point $x\in X$, $\Phi(x)=\{l_1,l_2\}$, is interpreted as if the data point $x$ has both of those properties. 

To understand the definition of a classifier, assume that the classifier $f$, associated with a labeling correspondence $\Phi$. Given a data point $x\subseteq X$,   $\Phi(x)$ is the set of all labels associated with the point $x$. To find the set of data points in the same class (or category) as $x$, we need to consider all data points with at least all the labels of the data point $x$. This is precisely the definition of $f(x)$.  

More generally, for a given set of data points $A\subseteq X$, $f(A)$ is the set of all data points that at least have all the labels in common with all points in $A$. The idea is that if a decision-maker wants to find all data points that are in the same class as the observed data points in $A$ (without any other information), she should consider all points $f(A)$.

\begin{remark}\label{rem:classifierisclosureoperator} A classifier derived from a labeling function is a closure operator: The extensivity and monotonicity properties are simple to verify. To show idempotence, notice that by monotonicity $f(A)\subseteq f(f(A))$. Hence, we only need to show that $f(f(A))\subseteq f(A)$. Assume that $x\in f(f(A))$. By definition, $\bigcap_{y\in f(A)} \Phi(y) \subseteq \Phi(x)$. We know that for every $y\in f(A)$ we have  $\bigcap_{z\in A} \Phi(z) \subseteq \Phi(y)$. Hence, $\bigcap_{z\in A} \Phi(z) \subseteq \bigcap_{y\in f(A)} \Phi(y) \subseteq \Phi(x) $. Thus, $x\in f(A)$. As a result, $f$ satisfies the idempotence property.
\end{remark}

The ideas may be illustrated by means of an example.

\begin{example}\label{ex:classifierexample}
Consider a set of four data points $X=\{a, b, c, d\}$. The set of labels is defined as $L=\{\text{dog}, \text{cat}, \text{black}, \text{white}, \text{female}, \text{male}, \text{car}\}$. Assume that the labeling correspondence $\Phi:X\rightrightarrows L$ is as follows: 

\begin{align*}
    \Phi(a)&=\{\text{dog}, \text{black}, \text{female}\} \\
    \Phi(b)&=\{\text{dog}, \text{black}, \text{male}\} \\
    \Phi(c)&= \{\text{cat}, \text{white}, \text{female}\} \\ 
    \Phi(d)&=\{\text{car}, \text{black}\}
\end{align*}

The classifier associated with the above labeling correspondence has eight classes. Class1=\{a\} associated with the labels \{\text{dog}, \text{black}, \text{female}\}, Class2=\{b\} associated with the labels \{\text{dog}, \text{black}, \text{male}\}, Class3=\{c\} associated with the labels \{\text{cat}, \text{white}, \text{female}\}, Class4=\{d\} associated with labels \{\text{car}, \text{black}\}, Class5=\{a,b\} associated with the labels \{\text{dog}, \text{black}\}, Class6=\{a,b,d\} associated with the labels
\{\text{black}\}, Class7=\{a,c\} associated with the labels \{\text{female}\}, and the last class is Class8=\{a,b,c,d\} associated with all data points in the set $X$.


Figure \ref{fig:examplemain} depicts the structure of the classifier and the associated lattice structure.

\begin{figure}
\centering
\begin{tikzpicture}
\matrix (A)[matrix of math nodes, column sep=0.6cm, row sep=0.5cm]{
 &  & \boxed{\{a,b,c,d\}}\\
&\boxed{\{a,b,d\}}\\
\boxed{\{a,b\}}& & &\boxed{\{a,c\}}\\
\boxed{\{a\}}& \boxed{\{b\}} &\boxed{\{c\}} &\boxed{\{d\}}\\
& & \boxed{\emptyset}\\};

\foreach \i/\j in {1-3/2-2, 1-3/3-4, 2-2/4-4, 2-2/3-1, 3-1/4-1, 3-1/4-2, 3-4/4-3, 3-4/4-1, 4-1/5-3, 4-2/5-3, 4-3/5-3, 4-4/5-3}
    \draw[<-] (A-\i) -- (A-\j);

\end{tikzpicture}

\caption{The lattice associated with the labeling correspondence $\Phi$ in Example \ref{ex:classifierexample}.}
\label{fig:examplemain}
\end{figure}
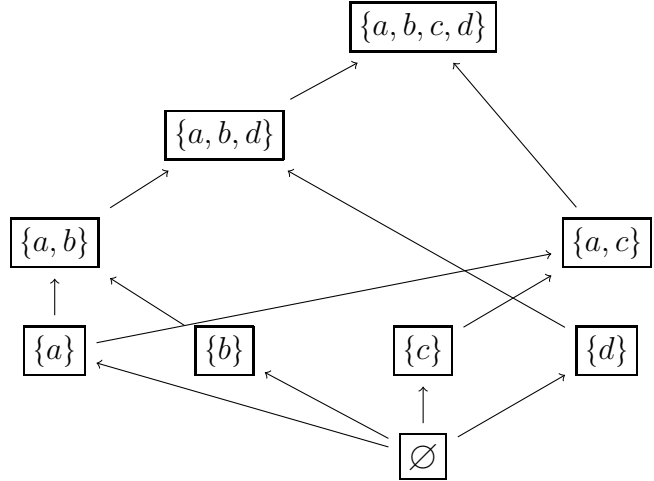


\end{example}

The topology $S(f)$, represents all possible different \df{classes}
of data points, and $f(A)$ is the smallest class that contains $A$. By
Lemma \ref{lem:closureintersection}, we know that the set
of classes is closed under intersections. Moreover, there is a
lattice structure associated with the labeling correspondence. Indeed,
we may interpret a classifier through the topology it induces: Let
$x\in X$ be in two classes $A, B \in S$, which means that it has the
properties of the classes $A$ and $B$. Then, there is a 
class $C\in S$, with the properties of the classes $A$ and $B$, such
that $x\in C$. The trivial class $X$ represents the set of all data
points, having every possible
label. 

These observations may be summed up in the following proposition.

\begin{proposition}\label{prop:labelingclassifier} Let $X$ be a set of data points. We have the followings:
\begin{enumerate}
    \item Let $\Phi:X\rightrightarrows L$ be a labeling correspondence. The classifier $f:2^X\to 2^X$ associated with $\Phi$ is a closure operator. 
    \item Let $f:X \to X$ be a closure operator on $X$. Then, there exists a set of labels $L$, and a labeling correspondence $\Phi: X\rightrightarrows L$ such that the classifier associated with $\Phi$ is $f$. Moreover, one choice of $L$ and $\Phi$ is achieved by defining $L=S(f)$ and $\Phi: X\rightrightarrows L$ with $\Phi(x)=\{s|\ s\in S(f), x\in s \}$.
\end{enumerate}
 \end{proposition}
The representation in  the second part of the above proposition  is not unique. However, section \ref{sec:complexity} provide the minimum number of labels needed to represent a given classifier.

\subsection{Application 2: Decision theory}

Consider a decision-maker who is going to ultimately choose an element
$x\in X$, but who is first presented with a choice among possible \df{menus}
$A\subseteq X$. After selecting the menu $A$, the consumer chooses an
alternative from $A$.

We may think of a consumer, for example, who
first  decides between retail shops, or restaurants, and then chooses an item from the 
shop, or a meal on the menu. But the situation is quite general: Most problems in dynamic choice involve
first selecting an alternative that will determine what may be chosen
in the future: think, for example, of choosing how much to consume and
how much to save, which will determine what is affordable in the
future. Or making a career choice, which may determine the set of
geographical location one may ultimately be able to choose among to
live in. We capture the situation in the abstract as a problem of
choosing a set in $X$. In particular, we shall focus on agents who
have a {\em preference for flexibility.}

Consider a decision maker, say Alice, who chooses according to a preference relation
$\succeq$ on $X=\{x,y,z \}$. If Alice can perfectly predict her choice, she will
only care about the best option available in a menu. So if, for
example, $x\succ y\succ z$ are Alice's preferences, then she will be indifferent between any
menu that contains $x$, and between any menu that does not contain $x$ but that contains $y$. Alice's situation may be described through a
strategically rational operator, with \[
\{x,y,z \} = f_{\succeq}(\{x\})= f_{\succeq}(\{x,y\}) = f_{\succeq}(\{x,z\}) = f_{\succeq}(\{x,y,z\}).
\]

Bob, in contrast, may be unsure about his ultimate choice from a
menu. He might have the same preference as Alice, or he might actually have the preference $y\succ' x\succ' z$ that
ranks $y$ first. Observe that, as a consequence, Bob will strictly
prefer the menu $\{x,y,z \}$ to $\{x\}$, or to $\{x,z\}$, as the larger menu does
not close any doors for Bob. If, in the end, Bob's preference is $\succeq'$,
then 
$\{x,y,z \}$ lets him get his favorite choice. Indeed for Bob we need
an operator with \[
\{x,y,z \} = f(\{x,y,z\}) = f(\{x,y\}) \neq \{x,z \} = f(\{x\}) = f(\{x,z\}).
\]

Bob's preferences cannot be captured through a single strategically
rational operator. Actually for him  we have that $f=f_{\succeq}\cap f_{\succeq'}$,
and thus Bob's operator $f$ is derived from {\em two} strategically rational
operators.

Now we see, as a preview of our results, that Bob's behavior is \emph{more
complex} than Alice's because we need two basic strategically
rational operators to capture Bob's choice, while one suffices for
Alice. 

In all, Bob has a ``preference for flexibility,'' a model that  was first proposed by
\cite{kreps}.\footnote{Kreps' model motivated a sizable literature in
  decision theory; see, for example, \cite*{DLR}, \cite*{DLRS}, \cite{GP},
  \cite{nehring}, \cite{nehringspuppe1999}, \cite{chateauneufandjaffray}, and \cite*{chambers2020closure}} In Section~\ref{sec:applicationDT} below we
briefly recap Kreps' model and show what our results have to say about
the resulting representation.

\section{A general representation of closure operators}\label{sec:genrep}

A key motivation for the study in our paper is the next result, due to \cite*{chambers2020closure}, which shows that a closure operator may be found as the intersection of supporting half-spaces.\footnote{Note also that \cite{richterandrubinstein} contains a similar result for linear orders.}

\begin{theorem}\label{prop:representationbyweakorders}
A function $f:2^X\to 2^X$ is a closure operator iff there exist weak orders $\succsim_1,\ldots,\succsim_k$ on $X$, such that
\begin{equation}\label{eq:absstractduality}
    f(A)=\bigcap_{i\in \{1,\ldots,k\}} H(\succsim_i,h_A).
\end{equation}
\end{theorem}

Using the language introduced above, we may rephrase this result as follows:

\begin{corollary}\label{cor:stratrat}
  A function $f:2^X\to 2^X$ is a closure operator iff there exist
  strategically rational operators $g_{u_i}$, $1\leq i\leq k$, so that 
\begin{equation}\label{eq:corstratrat}
    f(A)=\bigcap_{i\in \{1,\ldots,k\}} g_{u_i}(A).
\end{equation}
\end{corollary}

Our first result is an elaboration on Theorem~\ref{prop:representationbyweakorders}. It serves to introduce the notion that operators generate more complex operators, and how the resulting topologies are connected.

\begin{theorem}\label{thm:comborepresentation}
Let $g_1,\ldots,g_k$ be closure operators on $X$, then $f:2^X\to 2^X$ defined by \[    f(A)=\bigcap_{i\in \{1,\ldots,k\}} g_i(A) \] is a closure operator. We say that $f$ is \df{generated} from $g_1,\ldots,g_k$.

Moreover, if  $f, g_1,\ldots,g_k$ are closure operators on $X$ then $f$ is generated from $g_1,\ldots,g_k$ iff \begin{enumerate}
    \item $S(g_i)\subseteq S(f)$ for all $i\in \{1,\cdots,k\}$;\footnote{This statement is Theorem~5.1 in \cite{ward}.}
    \item and if $A\in S(f)$ and $x\notin A$, then there exists a closure operator $g_i\in \{g_1,\ldots,g_k\}$ such that $x\notin g_i(A)$.
\end{enumerate}
\end{theorem}

Theorem~\ref{thm:comborepresentation} warrants some remarks.
Assume that the classifier $f$ is generated by $g_1,\ldots,g_k$. By the first condition in the theorem, the topology $S(f)$ is finer than that of any constituent operator $S(g_i)$. 

The second condition is similar to the Separating Hyper Plane Theorem in convex geometry. It means that if $A$ is a closed set of $f$ and $x\notin A$, then there should be a separating classifier $g_i\in \{g_1,\ldots,g_k\}$ that detect that $x$ is not in the closure of $A$ with respect to $g_i$.

\section{Complexity of operators}\label{sec:complexity}

In the Machine Learning literature, a neural network is built from a
set of simple classifiers. Given a data set, and a set of labels, a
researcher can add many linear classifiers to build a large neural
network that can shatter data points to the correct classes. More
generally, a researcher can combine many different functions to form a
complex function with lots of parameters to shatter the set of data
points to the correct classes.

We proceed with an analogous motivation: in Theorem~\ref{thm:comborepresentation}, complex operators are built up from simpler ones. In the abstract, discrete, setting of our paper, there is no useful notion of norm or approximation, but the representation  in Theorem~\ref{thm:comborepresentation} serves as a starting point. When $f=\cap g_i$, we may think of $f$ as being more complex than any of its building blocks $g_i$, just as a neural network is more complex than its simpler constituents. In fact, using the theorem, we may associate complexity to the topology defined from each operator:

\begin{definition}
The operator $f$ is \df{more complex} than the operator $g$ if $S(g)\subseteq S(f)$.
  \end{definition}

So that a more complex operator induces a finer topology.  The ``more complex than'' relation is, we think natural, but it induces an incomplete partial order on operators, and will  render some pairs of operators incomparable. We consider two ways of completing the complexity binary relations.  

\begin{itemize}
\item The \df{minimum number of weak orders} (MNWO) of an operator $f$
  is the smallest number $n$ so there exists weak orders
  $\succeq_1,\ldots,\succeq_n$ with $f=\cap_{i=1}^n g_{\succeq_i}$.
\item The \df{minimum number of binary classifiers} (MNBC) of an
  operator $f$ is the smallest number $n$ so there exists subsets 
  $C_1,\ldots,C_n$ of $X$ with $f=\cap_{i=1}^n g_{C_i}$.
   \end{itemize}

The MNWO has a natural interpretation. We may think of strategically rational operators as simple because they do not exhibit a preference for flexibility. Or because they reflect a one-dimensional property (like size, color, being closer to the entrance of a supermarket). Their topologies are chains, which are naturally represented using $k-1$ labels when they are of length $k$. 

The MNBC may be though of as reflecting the length of the binary code needed to describe the topology of the closure operator. In fact, we can show that the minimum number of labels needed to describe a classifier, as a labeling correspondence, is exactly the same as minimum number of the binary classifier needed.

Let $P(f)$ by the elements of the topology $S(f)$ that are not the intersection of other closed sets in $S(f)$, and let $B(f)$ be the elements in $P(f)$ other than $\os$ and $X$.

\begin{proposition}\label{prop:complexityofoperators}  Let $f$ be a closure operator, then the MNWO of $f$ is equal to the  width of $P(f)$, and the MNBC is equal to the cardinality of $B(f)$.
\end{proposition}

\section{Application 1: complexity of classifiers}\label{sec:applicationclassifiers} 

In the application to classifiers, we may think of binary classifers $g_C$ as capturing a single property that datapoints may or may not posses, which makes them good candidates for the simplest possible classifers available, and thus suggest MNBC as a natural measure of complexity for a classifier.

In constrast, the strategically-rational operators $f_{\succeq}$ are
analogous to the linear classifers in Machine Learning.  Indeed,
$f_{\succeq}(A)=H(\succsim,h_A)$, a ``discrete halfspace,'' and the
topology associated to $f$ is
$S(f_\succsim)=\{H(A,h_A)|\ A\subseteq X\}$. Observe that 
$S(f_\succsim)$ is a single chain with respect to set inclusion (in
other words, the lattice associated with $f_\succsim$ is a total
order). Interestingly, the reverse is also correct. If a lattice associated with a classifier $f$ is a single chain (total order), then it is generated by a single unique weak order. 

\begin{proposition}\label{lem:weakorderlattice}
The lattice associated with a closure operator is a single chain if and only if a single weak order generates it. Moreover, the weak order generating the lattice is unique.
\end{proposition}

More generally, we say that $f$ is a \df{simple classifier with length $k$}, whenever the number of closed sets (classes) generated by $f$, $|S(f)|$, is $k+1$. Notice that $S(f)$ always contains $\emptyset$ and $X$.  Hence, a simple classifier with length $k$ has $k$ nonempty different classes. 

Consider the classifier associated with Example~\ref{ex:classifierexample}. There are eight classes other than $\emptyset$. Let $L=\{\text{Class1},\ldots,\text{Class8}\}$. Using the result of Proposition~\ref{prop:labelingclassifier}, one choice of labeling correspondence is
\begin{align*}
    \Phi(a)&=\{\text{Class1}, \text{Class5}, \text{Class6},  \text{Class7}, \text{Class8}\} \\
    \Phi(b)&=\{\text{Class2}, \text{Class5}, \text{Class6}, \text{Class8}\} \\
    \Phi(c)&= \{\text{Class3}, \text{Class7}, \text{Class8}\} \\
    \Phi(d)&=\{\text{Class4}, \text{Class6}, \text{Class8}\}.
\end{align*}
The above labeling correspondence is different that the original one in the example. However, they both have the same classifier with the same set of classes.

\begin{remark}\label{rem:weakorderclassifierlenghtk}
Using the proof of Lemma~\ref{lem:weakorderlattice}, any simple classifier with length $k$ is associated with a unique weak order with $k$ different indifference classes. Let $\{A_1,\ldots, A_k\}$ be a partition of $X$ into the indifference classes of a weak order $\succsim$ with $x\succ y$ whenever $y\in A_i, x\in A_j$ such that $i<j$. Then, $\{A_1,(A_1\cup A_2),\ldots,(A_1\cup A_2\cup\ldots\cup A_k)\}$ is the set of closed sets of a simple classifier with length $k$. The reverse can be done in the same way. 
\end{remark}

\begin{remark}
Consider a binary classifier $f_C:2^X\to2^X$. Assume that the
corresponding topology is $S(f)=\{\emptyset,C,X\}$. The
corresponding weak order $\succsim_f$ is $x\succsim y $ if and only if $x\in
A,y\in X\setminus A$. From the perspective of classification, there is
no difference between $f$ and another binary classifier $g$ with
$S(g)=\{\emptyset,X\setminus C,X\}$. The weak order associated with
$g$ is the reverse order of  $\succsim_f$. The classifier derived from
the label ``horse'' achieves the same purposes as one derived from ``not
a horse.''
\end{remark}

\begin{example}\label{ex:classifierintobinary}
By Theorem~\ref{prop:representationbyweakorders}, any given classifier $f$ may be generated by weak orders
$\succsim_1,\ldots,\succsim_k$. Now we use use
Theorem~\ref{thm:comborepresentation} to define binary classifiers $\{g_1,\dots,g_k\}$ that generate $f$. 
One simple way is to notice that for any binary classifier $g_i$,
$S(g_i)$ should be a subset of $S(f)$. Moreover, for any class
$\emptyset\neq A\in S(f)$ and $x\notin A$, there should be one of
$g_i$ to separate $x$ and $A$. Hence, if we consider all binary
classifiers $g_C$ for every closed set $C\in S(f)$, both requirements
of Theorem~\ref{thm:comborepresentation} will be satisfied.
\end{example}

In Example~\ref{ex:classifierintobinary}, the representation uses
$|S(f)|-1$ binary classifiers and is clearly not minimal. We shall see
how to obtain a minimal representation using binary classifiers, but
first we consider strategically rational classifers of different
length. One way to proceed (which we shall see is not optimal) is to
decompose $S(f)$ into chains:

\begin{example}\label{ex:decompositiontochains}
\sloppy Let $X=\{a,b,c\}$ be a set of data points and $f:2^X\to2^X$ be a closure operator with the set of closed sets $S(f)=\{\emptyset,\{a\},\{b\},\{c\},\{a,b\},\{a,b,c\}\}$. We consider the following three chains:
\begin{align*}
    S(g_1)&=\{\emptyset,\{a\},\{a,b\},\{a,b,c\}\}\\ S(g_2)&=\{\emptyset,\{b\},\{a,b\},\{a,b,c\}\}\\ S(g_3)&=\{\emptyset,\{c\},\{a,b,c\}\}
\end{align*}
Notice that, since both conditions of Theorem~\ref{thm:comborepresentation} are satisfied, then $\{g_1,g_2,g_3\}$ generates $f$. Figure~\ref{fig:exampleclosuretochain} illustrates the decomposition.

\begin{figure}
\centering

\begin{tikzpicture}[scale=0.1]
\matrix (A)[matrix of math nodes, column sep=0.6cm, row sep=0.5cm]{
&\boxed{\{a,b,c\}}\\
&\boxed{\{a,b\}}\\
\boxed{\{a\}}& \boxed{\{b\}} &\boxed{\{c\}}\\
& \boxed{\emptyset}\\};

\foreach \i/\j in {1-2/2-2, 1-2/3-3, 2-2/3-1, 2-2/3-2, 3-1/4-2, 3-2/4-2, 3-3/4-2}
    \draw[<-] (A-\i) -- (A-\j);

\end{tikzpicture}
\qquad
\begin{tikzpicture}[scale=.2]
\matrix (A)[matrix of math nodes, column sep=0.6cm, row sep=0.5cm]{
&\boxed{\{a,b,c\}}&\boxed{\{a,b,c\}}&\boxed{\{a,b,c\}}\\
&\boxed{\{a,b\}}&\boxed{\{a,b\}}\\
\boxed{\{a\}}& &\boxed{\{b\}}& &\boxed{\{c\}}\\
& \boxed{\emptyset}& \boxed{\emptyset}& \boxed{\emptyset}\\};

\foreach \i/\j in {1-2/2-2,1-3/2-3,1-4/3-5,2-2/3-1,3-1/4-2,3-3/4-3,2-3/3-3,3-5/4-4}
    \draw[<-] (A-\i) -- (A-\j);

\end{tikzpicture}

\caption{The lattice associated with the closure operator $f$ is in the top. The associated decomposition into $g_1,g_2,g_3$ is in the bottom.}
\label{fig:exampleclosuretochain}

\end{figure}
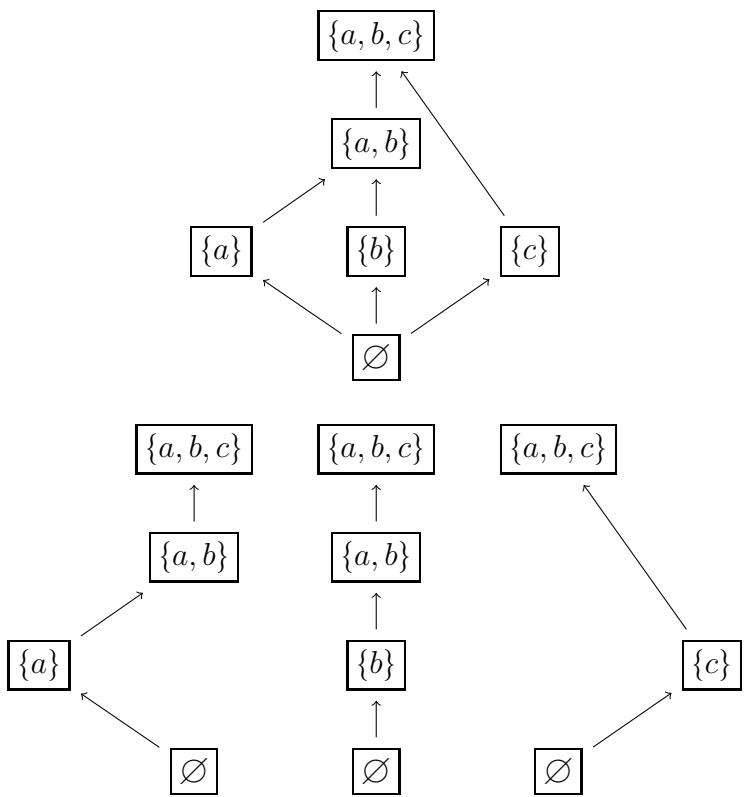
\end{example}

We take up a discussion of minimal representations in Section~\ref{sec:discussion}

\section{Application 2: Choice over menus}\label{sec:applicationDT}

Our second application is to the choice of menus, $A\subseteq
X$. In decision theory, the choice over menus are thought of as capturing the general problem of dynamic choice in an abstract setting. An agent is making choices today that may constrain her future choices. To this end, suppose that $\kpref$, a preference relation over $2^X$, captures choices made
over menus. The indifference relation derived from $\kpref$ is denoted by $\ksim$, so that $A\ksim B$ when $A\kpref B$ and $B\kpref A$. 

\cite{kreps} introduced the idea that an agent may have a preference for flexibility due to uncertainty about her ultimate future choices. He proposed some simple axioms on $\ksim$, and proved that they give rise to a particular kind of utility representation; a representation that reflects the agents' uncertainty over a state space that guides her future choices. 

The literature on menu choice in economics is substantial, but for our purposes we want to highlight the work of \cite{DLR} because their theory of choice hinges crucially on adopting a particular closure operator. In a setting of choice over lotteries, they use the convex hull operator (which is a closure operator in the Euclidean setting of their paper). Kreps' result also uses a closure operator in constructing his state space, but his axioms do not explicitly reference the closure operator.  

Following \cite{kreps}, we entertain the following axioms on $\kpref$:

\begin{enumerate}
    \item Desire for flexibility: $B\subseteq A$ implies $A\kpref B$, 
    \item Ordinal submodularity: $A\ksim A\cup B$ implies that for all $C$, $A\cup C\ksim A\cup B\cup C$.
\end{enumerate}

There are two possible approaches. The first is to define, as in
\cite{kreps}, the function  $f:2^X\to 2^X$ from preference $\kpref$ by
\begin{equation}\label{eq:krepsclosureeq}
f(A)=\bigcup_{B\in 2^X,\ A\ksim A\cup B} B .
\end{equation}

The second approach is proposed by  \cite{DLR}, using a convex-hull operator over subsets (menus) of a Euclidean space. The key property in this approach is that a closure operator is given so that $A\ksim f(A)$ for all menus $A$. 
\begin{definition}
 $\kpref$ \df{respects} $f$ if $A\ksim f(A)$ for every menu $A\in 2^X$.
\end{definition}

If we follow the first approach then  it is easy to show (see \cite{kreps}) that, under the two axioms,
\begin{enumerate}
    \item $f$ is a closure operator,
    \item $\kpref$ respects $f$,
    \item $A\ksim A\cup B$ if and only if $f(B)\subseteq f(A)$,
    \item $f(B)\subset f(A)$, then $A\succ' B$.
\end{enumerate}

In consequence, we obtain the following version of Kreps' result:

\begin{theorem}\label{prop:krepsdecomposition1} Let $\kpref$ be
  a preference relation over $2^X$ that satisfies desire for flexibility
  and ordinal submodularity, and let $f$ be defined using
  Equation~\eqref{eq:krepsclosureeq}. Then there is a function
  $U:X\times S(f) \to\mathbb{R}$, and a strictly increasing function
  $u:\mathbb{R}^S\to \mathbb{R}$ such that  \begin{equation}\label{eq:krepsfirst}
    u([\underset{a\in A}{max}\ U(a,s)]_{s\in S})
\end{equation} represents $\kpref$. The minimum number of states
  (cardinality of $S(f)$) needed for the representation is precisely
  the MNWO of the associated operator $f$, which is $|P(f)|$.
\end{theorem}

If we instead adopt a given closure operator $f$, we obtain an analogue to the result in \cite{DLR}:

\begin{theorem}\label{prop:additivegeneralrespect} Suppose given a
  closure operator $f$, and a preference $\kpref$ over $2^X$ that
  respects $f$. Then there exist a state space $S$, where $S=S^+\cup S^-$ with $S^+\cap S^-=\emptyset$ and has cardinality at most $2(|S(f)|-1)$, and a state-dependent utility $U:X\times S\to\mathbb{R}$, such that:
\begin{equation}\label{eq:additivelastresult}
U(A)=\sum_{s\in S^+}\underset{a\in A}{max}\, U(a,s)-\sum_{s\in S^-}\underset{a\in A}{max}\, U(a,s)
\end{equation}
represents $\kpref$.
\end{theorem}

Theorem~\ref{prop:additivegeneralrespect} presumes a closure operator
that respects the preference $\kpref$, but is otherwise
arbitrary. \cite{DLR}, working in a space of lotteries, impose that a
preference respects the convex-hull operator, mapping each set in a
Euclidean space into its convex hull (this is indeed a closure operator in
the Euclidean case). Our theorem shows that the convexity properties of lotteries are not needed. Once we adopt the framework in our paper, the same ideas may be extended to the purely finite and discrete setting of our paper.

Now, any preference relation is respected by the identity operator,
which gives rise to the following consequence of our result:

\begin{corollary}
For every preference ordering $\kpref$ over the set of menus, there exists a representation as in Equation~\ref{eq:additivelastresult} with at most $2\times(2^{|X|}-1)$ states. 
\end{corollary}

\begin{remark}
The representation constructed in the proof of
Theorem~\ref{prop:additivegeneralrespect} is not a minimal additive
representation. For example, consider a one-to-one utility function $U:X\to\mathbb{R}$, which induces a preference ordering $\kpref$ over the set of menus by means of \[ 
A \kpref B \text{ if and only if } \max\{U(x):x\in A \}\geq  \max\{U(x):x\in B \}.
\] It should be clear that the associated strategically rational operator has a topology of size $X$. Our construction in Proposition~\ref{prop:additivegeneralrespect}, however, generates a representation with at least $2(|X|-1)$ subjective states. 
\end{remark}

\section{Discussion}\label{sec:discussion}

We discuss the ideas behind our definitions of complexity through a
series of examples. This discussion will also lead up to a proof of
the minimal numbers of weak orders, and binary classifiers, needed to
represent a given closure operator.

We start with a simple example that illustrates the definitions in the paper:

\begin{example}\label{ex:exampleofcomplexity}
Let $X=\{a,b,c,d\}$ be the set of data points. Consider the classifiers $f_1,f_2$ defined by:
\begin{align*}
    S(f_1)&=\{\emptyset,\{a\},\{b\},\{a,b,c,d\}\}\\ S(f_2)&=\{\emptyset,\{a\},\{a,b\},\{a,b,c,d\}\}
\end{align*}

\begin{figure}
\centering
\begin{tikzpicture}[scale=0.3]
\matrix (A)[matrix of math nodes, column sep=0.6cm, row sep=0.5cm]{
&\boxed{\{a,b,c,d\}}\\
\boxed{\{a\}}& &\boxed{\{b\}}\\
& \boxed{\emptyset}\\};

\foreach \i/\j in {1-2/2-1, 1-2/2-3, 2-1/3-2, 2-3/3-2}
    \draw[<-] (A-\i) -- (A-\j);

\end{tikzpicture}
\qquad
\begin{tikzpicture}[scale=0.3]
\matrix (A)[matrix of math nodes, column sep=0.6cm, row sep=0.5cm]{
\boxed{\{a,b,c,d\}}\\
\boxed{\{a,b\}}\\
\boxed{\{a\}}\\
\boxed{\emptyset}\\};

\foreach \i/\j in {1-1/2-1, 2-1/3-1, 3-1/4-1}
    \draw[<-] (A-\i) -- (A-\j);

\end{tikzpicture}

\caption{The lattices associated with the closure operator $f_1$ (the left one) and $f_2$.}
\label{fig:exampletwobinary1}
\end{figure}
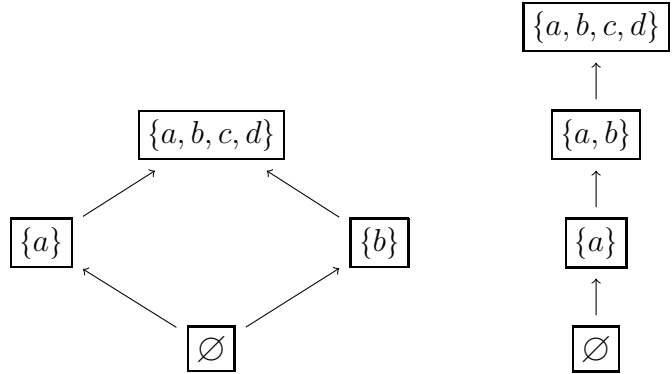

Figure~\ref{fig:exampletwobinary1}, illustrates the underlying lattice structures. 
The topology $S(f_1)$ has depth two and width two, while $S(f_2)$ has
depth three and width one. The number of
nonempty classes that $f_1$ or $f_2$ can detect is three. The MNWO of
$f_1$ is two, while that of $f_2$ is one.  The MNBC of both $f_1$ and
$f_2$ is two.

The calculations can be a lot more challenging. Consider the following classifiers $f_3$ and $f_4$:
\begin{align*}
    S(f_3)&=\{\emptyset,\{a\},\{a,b\},\{a,c\},\{a,b,c,d\}\}\\
    S(f_4)&=\{\emptyset,\{a\},\{a,b\},\{a,b,c\},\{a,b,c,d\}\}
\end{align*}

\begin{figure}
\centering

\begin{tikzpicture}[scale=0.3]
\matrix (A)[matrix of math nodes, column sep=0.6cm, row sep=0.5cm]{
&\boxed{\{a,b,c,d\}}\\
\boxed{\{a,b\}}& &\boxed{\{a,c\}}\\
&\boxed{\{a\}}\\
& \boxed{\emptyset}\\};

\foreach \i/\j in {1-2/2-1, 1-2/2-3, 2-1/3-2, 2-3/3-2, 3-2/4-2}
    \draw[<-] (A-\i) -- (A-\j);

\end{tikzpicture}
\qquad
\begin{tikzpicture}[scale=0.3]
\matrix (A)[matrix of math nodes, column sep=0.6cm, row sep=0.5cm]{
\boxed{\{a,b,c,d\}}\\
\boxed{\{a,b,c\}}\\
\boxed{\{a,b\}}\\
\boxed{\{a\}}\\
\boxed{\emptyset}\\};

\foreach \i/\j in {1-1/2-1, 2-1/3-1, 3-1/4-1, 4-1/5-1}
    \draw[<-] (A-\i) -- (A-\j);

\end{tikzpicture}

\caption{The lattices associated with the closure operator $f_3$
  (left) and $f_4$ (right).}
\label{fig:exampletwobinary2}
\end{figure}
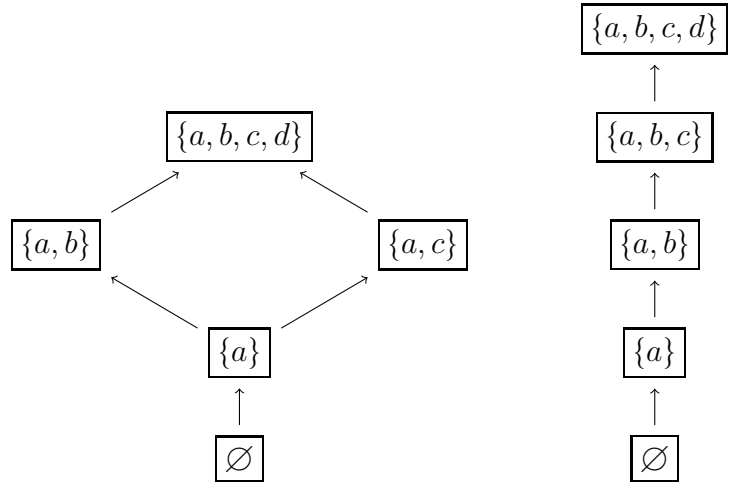

Figure~\ref{fig:exampletwobinary2}, illustrates the associated
lattices. The width, depth, and the number of classes can easily be
seen. The classifier $f_3$ has depth three, width two, and four
nonempty classes. Similarly, the classifier $f_4$ has depth four,
width one, and four nonempty classes.  Two simple classifiers can
generate the classifier $f_3$, and the classifier $f_1$ can be
generated by one simple classifier. 

However, when it comes to the decomposition into the binary
classifiers, we need to be more careful. The classifier $f_4$ can be generated with three binary classifiers. It is not hard to see that the following classifiers are the only minimal representations of $f_4$ through binary classifiers:
\begin{align*}
    S(f_{4,1})&=\{\emptyset,\{a\},\{a,b,c,d\}\}\\
    S(f_{4,2})&=\{\emptyset,\{a,b\},\{a,b,c,d\}\}\\
    S(f_{4,3})&=\{\emptyset,\{a,b,c\},\{a,b,c,d\}\}
\end{align*}

Now consider $f_3$. We might, at first, think that we need at least
three binary classifiers to generate $f_3$. However, perhaps surprisingly, in this case, we only need two binary classifiers. We define the classifiers $f_{3,1},f_{3,2}$ as follows:
\begin{align*}
    S(f_{3,1})&=\{\emptyset,\{a,b\},\{a,b,c,d\}\}\\
    S(f_{3,2})&=\{\emptyset,\{a,c\},\{a,b,c,d\}\}
\end{align*}

By the second condition in Proposition~\ref{thm:comborepresentation}, since $\{a\}=\{a,b\}\cap\{a,c\}$, we can see that $f_{3,1},f_{3,2}$ generate $f_3$. \emph{Hence, if we think about MNBC as a notion of complexity, then $f_4$ is more complex than $f_3$.} 
\end{example}

\begin{remark}\label{rmk:widthMNSC}
  The MNWO of an operator $f$ is bounded by the width of the topology
  $S(f)$, but may be strictly smaller.
\end{remark}
Our next example illustrates Remark~\ref{rmk:widthMNSC}.

\begin{example}\label{ex:widthisnotmin}
Let $X=\{a,b,c\}$ be a set of data points. We define the classifier $f$ as follows:

$$S(f)=\{\emptyset , \{a\}, \{b\}, \{c\}, \{a,b\}, \{b,c\},\{a,b,c\}\}$$

\begin{figure}
\centering
\begin{tikzpicture}
\matrix (A)[matrix of math nodes, column sep=0.6cm, row sep=0.5cm]{
 &  & \boxed{\{a,b,c\}}\\
&\boxed{\{a,b\}}&&\boxed{\{b,c\}}\\
\boxed{\{a\}}& &\boxed{\{b\}}& &\boxed{\{c\}}\\
& & \boxed{\emptyset}\\};

\foreach \i/\j in {1-3/2-2, 1-3/2-4, 2-2/3-1, 2-2/3-3, 2-4/3-3, 2-4/3-5, 3-1/4-3, 3-3/4-3, 3-5/4-3}
    \draw[<-] (A-\i) -- (A-\j);

\end{tikzpicture}

\caption{The lattice associated with the classifier $f$.}
\label{fig:examplewidthnotmin}
\end{figure}
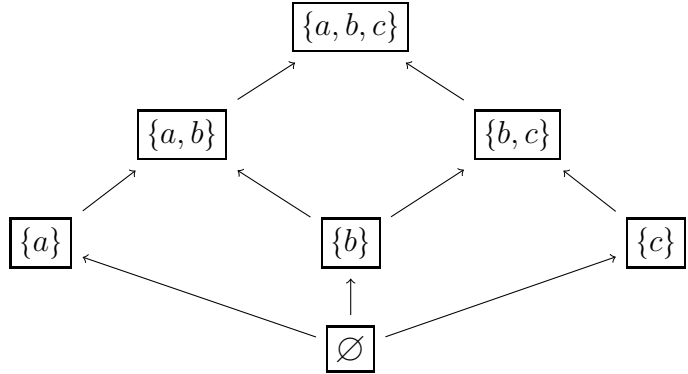

Figure~\ref{fig:examplewidthnotmin} illustrates the lattice associated with $f$. The antichain with classes $\{a\},\{b\},\{c\}$ is the largest antichain in the lattice. Therefore, the width of the lattice $S(f)$ is three. As a result of the Dilworth's Theorem, one can decompose the lattice to three chains. For example, the three chains $C_1=\{\emptyset,\{a\},\{a,b\},\{a,b,c\}\}$, $C_2=\{\{c\},\{b,c\}\}$, and $C_3=\{\{b\}\}$. Attaching $\emptyset$ and $\{a,b,c\}$ to each of the three chains implies the following three simple classifiers:

\begin{align*}
    S(f_{C_1})&=\{\emptyset,\{a\},\{a,b\},\{a,b,c\}\}\\
    S(f_{C_2})&=\{\emptyset,\{c\},\{b,c\},\{a,b,c\}\}\\
    S(f_{C_3})&=\{\emptyset,\{b\},\{a,b,c\}\}
\end{align*}

However, we can generate $f$ by only two simple classifiers:
\begin{align*}
    S(f_1)&=\{\emptyset,\{a\},\{a,b\},\{a,b,c\}\}\\
    S(f_2)&=\{\emptyset,\{c\},\{b,c\},\{a,b,c\}\}
\end{align*}

By using the second condition of
Proposition~\ref{thm:comborepresentation}, since
$\{b\}=\{a,b\}\cap\{b,c\}$ then $f_1,f_2$ can generate $f$. Therefore,
the MNWO is only two, while the width of $S(f)$ is three. 
\end{example}

MNWO may be smaller than width when a class is the same as the
intersection of some other classes detected by some simple
classifiers. The first class may be detected without needing to
explicitly add it to a simple classifier.

\begin{example}\label{ex:revisitepriovousexample}
Consider Example~\ref{ex:widthisnotmin}. The set $P(f)$ is obtained by removing the classes $\emptyset$ and $\{b\}$ from $S(f)$. Figure~\ref{fig:posetobtained} illustrates the structure of $P(f)$. The width of $P(f)$ is two, which is the same as the MNWO of $f$.

\begin{figure}
\centering
\begin{tikzpicture}
\matrix (A)[matrix of math nodes, column sep=0.6cm, row sep=0.5cm]{
 &  & \boxed{\{a,b,c\}}\\
&\boxed{\{a,b\}}&&\boxed{\{b,c\}}\\
\boxed{\{a\}}& && &\boxed{\{c\}}\\
& & \\};

\foreach \i/\j in {1-3/2-2, 1-3/2-4, 2-2/3-1, 2-4/3-5}
    \draw[<-] (A-\i) -- (A-\j);

\end{tikzpicture}

\caption{The lattice associated with $P(f)$.}
\label{fig:posetobtained}
\end{figure}
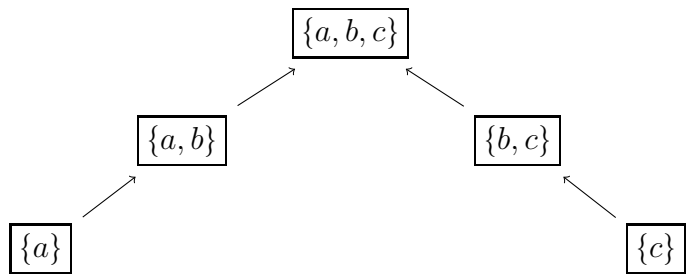
\end{example}

\section{Proofs}\label{sec:proofs}

\subsection{Proof of Proposition~\ref{prop:labelingclassifier}.}\label{sec:proofofproplabelingclassifier}

We have already proved the first part in remark \ref{rem:classifierisclosureoperator}. 
To prove the second part, we define $L=S(f)$ and $\Phi: X\rightrightarrows L$ with $\Phi(x)=\{s|\ s\in S(f), x\in s \}$. Based on the proof of the first part, we know that the classifier $g:2^X\to 2^X $ associated with $\Phi$ is a closure operator. Thus, we only need to prove that $g=f$.

First, we prove that $g(A)\subseteq f(A)$ for every $A\in 2^X$. Let $A\in 2^X$ and $x\in g(A)$. By the definition of $g$,  $\bigcap_{y\in A}\Phi(y)\subseteq \Phi(x)$. Then, by definition of $\Phi$, we have $\{ s\in S(f)|\  \forall y\in A, y\in s \}\subseteq \{s\in S(f)| \ x\in s \}$. Hence, if $s\in S(f)$ with $A \subseteq s$, then $x\in s $. 

Now, consider the set $s=f(A)$. By monotonicity of the closure operator we have $A\subseteq s$. Since $s\in S(f)$ and $A \subseteq s$, then we obtain that $x\in s$. This means that $x\in f(A)$. Thus, $g(A)\subseteq f(A)$ for every $A\in 2^X$.

For the other side, we need to show that $f(A)\subseteq g(A)$ for every $A\in 2^X$. Let $A\in 2^X$ and $x\in f(A)$. By the definition of $\Phi$ and $g$, it remains to show that $\bigcap_{y \in A} \Phi(y)=\{s\in S(f)| A\subseteq s \}\subseteq \Phi(x)=\{ s\in S(f)| \ x\in s \}$. Thus, it remains to show that if $s\in S(f)$ and $A\subseteq s$, then we would have $x\in s$. 
However, $f$ is monotonic and $s$ is a closed set respect to $f$. Thus, since $A\subseteq s$ then, $f(A)\subseteq f(s)=s$. 
Moreover, we assumed that $x\in f(A)$. As a result, we have $x\in s$.

\subsection{Proof of Theorem~\ref{thm:comborepresentation}}
Suppose first that $f(A)=\bigcap_{i\in \{1,\ldots,k\}} g_i(A)$. 
Extensivity and monotonicity of $f$ come from extensivity and monotonicity of each of $g_1,\ldots,g_k$. 
For the idempotence property, we need to show that $f(f(A))=f(A)$. By monotonicity of $f$, we only need to show that $f(f(A))\subseteq f(A)$.

Let $x$ be in $f(f(A))$. By the definition of $f$, for all $i\in \{1,\ldots,k\}$ $x\in g_i(f(A))$. Again by the definition of $f$, for all $i\in \{1,\ldots,k\}$ $x\in g_i(\bigcap_{j\in \{1,\ldots,k\}} g_j(A))$. Since $\bigcap_{j\in \{1,\ldots,k\}} g_j(A)\subseteq g_i(A)$ and $g_i$ is monotonic, then $g_i(\bigcap_{j\in \{1,\ldots,k\}} g_j(A))\subseteq g_i(g_i(A))$. But, since $g_i$ is a closure operator, then $g_i(g_i(A))=g_i(A)$. As a result, $g_i(\bigcap_{j\in \{1,\ldots,k\}} g_j(A))\subseteq g_i(A)$. Thus, for all $i\in \{1,\ldots,k\}$ we have $x\in g_i(A)$, which means that $x\in \bigcap_{j\in \{1,\ldots,k\}} g_j(A)$. Again by the definition of $f$, the last result shows that $x\in f(A)$. Hence, $f(f(A))\subseteq f(A)$, which completes the proof.

Now we turn to the second statement in the theorem. So let $f, g_1,\ldots,g_k$ be  closure operators on $X$. First, we show that if $f$ is generated by closure operators $g_1,\ldots,g_k$, then the two conditions in the theorem hold.

To prove the first condition, observe that $f$ is a closure operator
by the first part of the proof. Let $g_i\in\{g_1,\ldots,g_k\}$ and $A\in S(g_i)$. Since every $g_j\in \{g_1,\ldots,g_k\}$ is a closure operator, then $g_j$ is monotonic. Hence, for every $g_j$, $A\subseteq g_j(A)$. Moreover, since $A\in S(g_i)$, then $g_i(A)=A$. As a result, $A=\bigcap_{j\in \{1,\ldots,k\}} g_j(A)$, which means that $A=f(A)$. Thus, $A\in S(f)$. Hence, for every $i\in \{1,\dots,k\}$ we have $S(g_i)\subseteq S(f)$.

To prove the second condition, let $A\in S(f)$ and $x\notin A$. By the definition of $f$, $x\notin A=f(A)=\bigcap_{j\in \{1,\ldots,k\}} g_j(A)$. Therefore, there exists some $i\in \{1,\ldots,k\}$ such that $x\notin g_i(A)$. Thus, we complete the proof of the second condition.

Conversely, suppose now that  both conditions are satisfied. We proceed to show that $f(A)=\bigcap_{j\in \{1,\ldots,k\}} g_j(A)$.

By the first condition, for all $i\in \{1,\ldots,k\} $ $S(g_i)\subseteq S(f)$. Hence, for every $A\in 2^X$ and for every $i\in \{1,\ldots,k\}$ we have $A\subseteq g_i(A)\subseteq f(A)$. Therefore, $A\subseteq \bigcap_{j\in \{1,\ldots,k\}} g_j(A)\subseteq f(A)$. Since $S(f)$ is closed under intersection and $g_i(A)\subseteq\ S(g_i)\subseteq S(f)$, we also have $\bigcap_{j\in \{1,\ldots,k\}} g_j(A)\in S(f)$.

Since $f$ is a closure operator, it is monotonic. Applying $f$ to all terms of $A\subseteq \bigcap_{j\in \{1,\ldots,k\}} g_j(A)\subseteq f(A)$, gives us $f(A)\subseteq f(\bigcap_{j\in \{1,\ldots,k\}} g_j(A))\subseteq f(f(A))$. By idempotence property of $f$, we have $f(A)\subseteq f(\bigcap_{j\in \{1,\ldots,k\}} g_j(A))\subseteq f(A)$, which results in $f(A)= f(\bigcap_{j\in \{1,\ldots,k\}} g_j(A))$. But we have already shown that $\bigcap_{j\in \{1,\ldots,k\}} g_j(A)\in S(f)$. Therefore, $f(A)=f(\bigcap_{j\in \{1,\ldots,k\}} g_j(A))=\bigcap_{j\in \{1,\ldots,k\}} g_j(A)$. Thus, we show that $f(A)=\bigcap_{j\in \{1,\ldots,k\}} g_j(A)$. The last result completes the proof.

\subsection{Proof of Proposition~\ref{prop:complexityofoperators}}

Let a closure operator $f$ be generated by a set of closure operators $\{g_{\succsim_1}, \ldots, g_{\succsim_n}\}$. Consider a closed set $A\in S(f)$. If $A$ is not an intersection of other elements of $S(f)$, then by Theorem~\ref{thm:comborepresentation} there should be at least one $g_{\succsim_i}$ such that $A\in S(g_{\succsim_i})$. Moreover, if a set of closure operators $\{g_{\succsim_1}, \ldots, g_{\succsim_n}\}$ has the following two properties. First, $S(g_{\succsim_i})\subseteq S(f)$. Second, they can detect all the closed sets of $S(f)$ that are not the intersections of other closed sets of $S(f)$. Then, $f$ can be generated by $\{g_{\succsim_1}, \ldots, g_{\succsim_n}\}$.

As a result, we remove all the classes that are the intersections of some other classes from the lattice $S(f)$. In other words, in the lattice, $S(f)$, every class with an out-degree of more than two will be removed. The remaining set is a partially ordered set with respect to the set inclusion. Note that different weak orders should generate the elements of any anti-chain. However, by Dilworth's Theorem, the reverse is also true; the minimum number of weak orders that can cover the lattice is the length of the largest anti-chain. Therefore, the MNWO should be the width of the remaining partially ordered set.

The same observation and technique can be applied to find the minimum number of binary classifiers needed to generate a given closure operator. However, since both $\emptyset$ and $X$ are generated through any binary classifier, we can remove them from $S(f)$. Every other class in the remaining partially ordered set $B(f)$ should be contained in one of the binary classifiers.

\subsection{Proof of Proposition~\ref{lem:weakorderlattice}}
If $\succsim$ generates $f$, then Theorem~\ref{prop:representationbyweakorders} guaranties that $f$ is a closure operator. By the definition of $f$, the set of closed sets is $S(f_\succsim)=\{H(\succsim,h_A)|\ A\subseteq X\}$. Using the definition of support half-space, $S(f_\succsim)$ is a single chain respect to the set inclusion.

For the other side. If $f$ is a closure operator such that $S(f)$ is a single chain, then define $a\succsim b$ if and only if $f(a)\supseteq f(b)$. Since $S(f)$ is a total order, $\succsim$ is a weak order.

To show that $\succsim$ generates $f$, we need to show that for every $A\in 2^X$ we have $f(A)=H(\succsim,h_A)$.
Notice that there should be some $x\in A$ with $f(x)=f(A)$. Otherwise, since $S(f_\succsim)$ is a single chain, there should be a proper subset of $f(A)$ which contains all the closure of the singleton subsets of $A$, which is not correct. 

Now, consider any $y\in f(A)$. Since $f(y)\subseteq f(x)$, then $x\succsim y$. Thus, $y\in H(\succsim,h_A)$. As a result, $f(A)\subseteq H(\succsim,h_A)$. 

For the other side, since $f(x)=f(A)$, then $x$ should be a maximal element in $H(\succsim,h_A)$. Hence,  for all $y\in H(\succsim,h_A)$ we have $f(y) \subseteq f(x)=f(A)$. As a result $H(\succsim,h_A)\subseteq f(A)$. The last result completes the proof.

\subsection{Proof of Theorem~\ref{prop:additivegeneralrespect} }

We use the M\"{o}bius inversion formula to prove the result. Appendix~\ref{appendix:mobius} explains the technique in more detail. 

Consider the lattice $S(f)$. Define the partial order $\succsim$ over $S(f)$ by reversing the partial order induced by the set inclusion. In other words, $A\succsim B$ if and only if $A\subseteq B$. We can check that the meet and join of the lattice $S(f)$ will be swapped.

Since $\kpref$ is a transitive binary relation, we can extend it to a weak order over the finite set $2^X$. Then there should be a representation by real-valued functions. Consider any utility function $U:2^X\to\mathbb{R}$ that represents $\kpref$. 

We define the M\"{o}bius operator $\Phi:(S(f))^\mathbb{R}\to (S(f))^\mathbb{R}$ as follows:
\begin{equation}
    \Phi(f)(A)=\sum\limits_{\substack{A\succsim B \\ B\in S(f)}} U(B).
\end{equation}

M\"{o}bius inversion formula guarantees that the M\"{o}bius operator is bijective and the inverse is $\Phi^{-1}(g)(A)=\sum_{A\succsim B} \mu(B,A)g(B)$, where $\mu$ is the M\"{o}bius function.

As a result, if we define the function $h:S(f)\to \mathbb{R}$, for every $A\in S(f)$ as:
\begin{equation}
    h(A)=\sum\limits_{\substack{A\succsim B \\ B\in S(f)}} \mu(B,A)U(B) ,
\end{equation}

Then for every $A\in S(f)$, $U$ can be retrieved as follows:
 \begin{equation}
     U(A)=\sum\limits_{\substack{A\succsim B \\ B\in S(f)}}h(B).
 \end{equation}
 
However, $A\succsim B$ if and only if $A\subseteq B$. Therefore, we have
\begin{equation}
      U(A)=\sum\limits_{\substack{A\subseteq B \\ B\in S(f)}}h(B).
\end{equation}

Notice that the above equation is only correct for $A\in S(f)$.  However, since $\kpref$ respects $f$ and $U$ represents $\kpref$, then for every $A\in 2^X$ we have $U(A)=U(f(A))$. Therefore, for every $A\in 2^X$, since $f(A)\in S(f)$, we have
\begin{equation}
      U(A)=U(f(A))=\sum\limits_{\substack{f(A)\subseteq B \\ B\in S(f)}}h(B).
\end{equation}

Note that, since f is a closure operator, then $A\subseteq B$ if and only if $f(A)\subseteq B$ for every $A\in 2^X$ and every $B\in S(f)$. Therefore, for every $A\in 2^X$ we have

\begin{align}
     U(A)=\sum\limits_{\substack{A\subseteq B \\ B\in S(f)}}h(B).
\end{align}

We define $h^+(B)=max(0,h(B))$ and $h^-(B)=max(0,-h(B))$. Since $h=h^+-h^-$, then we have

\begin{equation}\label{eq:easyrep}
      U(A)=\sum\limits_{\substack{A\subseteq B \\ B\in S(f)}}h^+(B)-\sum\limits_{\substack{A\subseteq B \\ B\in S(f)}}h^-(B).
\end{equation}
 
By comparing the above equation and Equation~\ref{eq:additivelastresult}, we only need to make some changes to make them equal. The trick is as follows.

We define functions $U^+,U^-:(X\times S(f)\setminus\emptyset) \to \mathbb{R}$ as follows:

\begin{align*}
    U^+(x,B)&=
    \begin{cases}
    -h^+(B) & \text{if}\  x\in B \\
    0 & \text{if}\  x\notin B 
    \end{cases}
    \\
     U^-(x,B)&=
    \begin{cases}
    -h^-(B) & \text{if}\  x\in B \\
    0 & \text{if}\  x\notin B 
    \end{cases}
\end{align*}

Now, consider any $A\in 2^X$ and $B\in S(f)\setminus\emptyset$. By our definition of $U^+,U^-$, we have

\begin{align*}
    \underset{a\in A}{max}\ U^+(a,B)&=
    \begin{cases}
     -h^+(B) & \text{if}\  A\subseteq B \\
    0 & \text{otherwise}
    \end{cases}
    \\
    \underset{a\in A}{max}\ U^-(a,B)&=
    \begin{cases}
     -h^-(B) & \text{if}\  A\subseteq B \\
    0 & \text{otherwise}
    \end{cases}
\end{align*}

As a result of the above observation, we get the following result:
\begin{align}\label{eq:easyfinalrep}
      U(A)&=\sum\limits_{\substack{A\subseteq B  \\ B\in S(f)}}h^+(B)-\sum\limits_{\substack{A\subseteq B \\ B\in S(f)}}h^-(B) \notag\\
      &=(\sum\limits_{\substack{A\subseteq B \\ B\in S(f)\setminus\emptyset}}-\underset{a\in A}{max}\ U^+(a,B))-(\sum\limits_{\substack{A\subseteq B \\ B\in S(f)\setminus\emptyset}}-\underset{a\in A}{max}\ U^-(a,B)) \notag\\
      &=(\sum\limits_{B\in S(f)\setminus\emptyset}-\underset{a\in A}{max}\ U^+(a,B))-(\sum\limits_{B\in S(f)\setminus\emptyset}-\underset{a\in A}{max}\ U^-(a,B))\notag \\
      &=-(\sum\limits_{B\in S(f)\setminus\emptyset}\underset{a\in A}{max}\ U^+(a,B))+(\sum\limits_{B\in S(f)\setminus\emptyset}\underset{a\in A}{max}\ U^-(a,B))\notag \\
      &=(\sum\limits_{B\in S(f)\setminus\emptyset}\underset{a\in A}{max}\ U^-(a,B))-(\sum\limits_{B\in S(f)\setminus\emptyset}\underset{a\in A}{max}\ U^+(a,B))
\end{align}

Equation \ref{eq:easyfinalrep} and  \ref{eq:additivelastresult} are similar except their indexes. To make them the same, we consider any two disjoint sets $S^+,S^-\subseteq \mathbb{N}$~\footnote{The choice of $\mathbb{N}$ is arbitrary. As long as $S^+,S^-$ are disjoint and each has $|S(f)|-1$ elements, our argument follows.}, with both have $|S(f)|-1$ elements. We consider any two bijection $index_1:S^+\to S(f)\setminus\emptyset,\ index_2:S^-\to S(f)\setminus\emptyset$. Let $S=S^+\cup S^-$. We define the function $U:X\times S\to\mathbb{R}$ as follows:

\begin{align*}
    U(x,s)&=\begin{cases}
    U^-(x,index_1(s)) &\text{if} \ s\in S^+\\
    U^+(x,index_2(s)) &\text{if} \ s\in S^-
    \end{cases}
\end{align*}

Then, using Equation~\ref{eq:easyfinalrep} and our definition of function $U$, we have the following result:

\begin{align}\label{eq:finalfinal}
    U(A)&=(\sum\limits_{B\in S(f)\setminus\emptyset}\underset{a\in A}{max}\ U^-(a,B))-(\sum\limits_{B\in S(f)\setminus\emptyset}\underset{a\in A}{max}\ U^+(a,B))\notag\\
    &=\sum_{s\in S^+}\underset{a\in A}{max}\ U(a,s)-\sum_{s\in S^-}\underset{a\in A}{max}\ U(a,s)
\end{align}
Equation~\ref{eq:finalfinal} finishes our proof.


\nocite{*}

\printbibliography


\appendix

\section{Appendix}
\subsection{M\"{o}bius Operator}\label{appendix:mobius}
The main technical tool in our results on menu choice is the M\"{o}bius inversion formula\footnote{For an application of the technique in Kreps' setting, check \cite{nehring}. For a complete study of the concept, see \cite{chateauneufandjaffray}.}. Let $(X,\succsim)$ be a finite partially ordered set. We define the \text{\emph{M\"{o}bius}} function $\mu:X\times X\to \mathbb{R}$ as follows: 

We set $\mu(x,y)=0$ whenever $y\nsuccsim x$, and $\mu(x,x)=1$ for all $x\in X$.  Then, define other values inductively as follows:
$$\mu(x,y)=-\sum_{y\succ z\succsim x} \mu(x,z).$$

By above definition, we have:
\begin{equation*}
\sum_{y\succsim z\succsim x} \mu(x,z) =
\begin{cases}
1 & \text{if } x=y,\\
0 & \text{otherwise}.
\end{cases}
\end{equation*}

Let $(X)^\mathbb{R}$ denote the set of all functions from $X$ to $\mathbb{R}$. Then, the \textbf{\emph{M\"{o}bius operator}} $\Phi:(X)^\mathbb{R}\to (X)^\mathbb{R}$ is defined by $\Phi(f)(x)=\sum_{x\succsim y} f(y)$. M\"{o}bius inversion formula guarantees that the M\"{o}bius operator is bijective and the inverse is $\Phi^{-1}(g)(x)=\sum_{x\succsim y} \mu(x,y)g(y)$.

\begin{theorem}(M\"{o}bius inversion formula)
Let $(X,\succsim)$ be a finite partially ordered set and $\mu$ be its M\"{o}bius function. Let $f,g:X\to\mathbb{R}$. Then
$$g(x)=\sum_{x\succsim y} f(y)$$
implies that
$$f(x)=\sum_{x\succsim y} \mu(y,x)g(y).$$

\end{theorem}

\begin{proof}
 To be complete, we add the proof. 
 \begin{align*}
     \sum_{x\succsim y} \mu(y,x)g(y)&=\sum_{x\succsim y} \mu(y,x)\big( \sum_{y\succsim z} f(z)\big)\\
     &=\sum_{x\succsim y\succsim z} \mu(y,x)f(z)\\
     &=\sum_{x\succsim z}\big(\sum_{x\succsim y \succsim z} \mu(y,x)\big)f(z)=f(x).
 \end{align*}
 \end{proof}

\end{document}
\endinput